\documentclass[aps,pra,twocolumn,showpacs,superscriptaddress,reprint]{revtex4-1}

\usepackage{amsmath}
\usepackage{amssymb}
\usepackage{amsthm}
\usepackage{nicefrac}
\usepackage{tabularx}
\usepackage{hyperref}
\usepackage{url}
\usepackage{dsfont}
\usepackage{xspace}
\usepackage{setspace}

\usepackage{array}

\newcolumntype{C}{>{$}c<{$}}
\newcolumntype{L}{>{$}l<{$}}
\newcolumntype{R}{>{$}r<{$}}

\hypersetup{colorlinks=true,citecolor=cyan,urlcolor=blue}

\newtheorem{defi}{Definition}
\newtheorem{remark}{Remark}
\newtheorem{lemma}{Lemma}
\newtheorem{corollary}{Corollary}
\newtheorem{theorem}{Theorem}

\newcommand{\eqrf}[2][\empty]{\ensuremath{\stackrel{\text{\rm #1 \ref{#2}}}{=}}}
\newcommand{\neqrf}[2][\empty]{\ensuremath{\stackrel{\text{\rm #1 \ref{#2}}}{\neq}}}

\def\|#1>{\left|#1\right\rangle}
\def\<#1|{\left\langle#1\right|}

\newcommand{\osqr}[1][1]{\tfrac{#1}{\sqrt{2}}}

\renewcommand{\epsilon}{\varepsilon}
\newcommand{\eps}{\epsilon_{abc}}

\newcommand{\omk}{\omega_k}

\newcommand{\rmi}{\mathrm{i}}

\newcommand{\numberset}[1]{\mathds{#1}}

\newcommand{\Z}{\numberset{Z}}
\newcommand{\N}{\numberset{N}}

\newcommand{\C}{\numberset{C}}
\newcommand{\Q}{\numberset{Q}}
\newcommand{\F}{\numberset{F}}

\newcommand{\lmat}[4]{\begin{pmatrix} #1 & #2 \\ #3 & #4 \end{pmatrix}}

\newcommand{\set}[2]{\{{#1}\mid{#2}\}}

\newcommand{\pr}[2][\empty]{\ifx#1\empty\sqrt{\sigma_{#2}}\else\sqrt[#1]{\sigma_{#2}}\fi}

\newcommand{\ur}[2][\empty]{\ifx#1\empty V_{2,#1}\else V_{#1,#2}\fi}

\DeclareMathOperator{\lcm}{lcm}
\DeclareMathOperator{\Tr}{Tr}

\def\<#1>{\left\langle #1 \right\rangle}

\begin{document}

\title{Translating between the roots of the identity in quantum computers}
\author{Wouter Castryck}
\author{Jeroen Demeyer}
\affiliation{Vakgroep wiskunde, Universiteit Gent, Gent, Belgium}
\author{Alexis De Vos}
\affiliation{Cmst, Imec v.z.w.\ and vakgroep elektronica en informatiesystemen,
  Universiteit Gent, Gent, Belgium}
\author{Oliver Keszocze}
\author{Mathias Soeken}
\affiliation{Faculty of Mathematics and Computer Science,
  University of Bremen, Bremen, Germany}
\affiliation{Cyber-Physical Systems, DFKI GmbH, Bremen, Germany}


\begin{abstract}
  The Clifford+$T$ quantum computing gate library for single qubit
  gates can create all unitary matrices that are generated by the group
  $\<H, T>$.  The matrix $T$ can be considered the fourth root of Pauli $Z$,
  since $T^4 = Z$ or also the eighth root of the identity $I$.  The Hadamard
  matrix $H$ can be used to \emph{translate} between the Pauli matrices, since
  $(HTH)^4$ gives Pauli $X$.  We are generalizing both these roots of the Pauli
  matrices (or roots of the identity) and translation matrices to investigate
  the groups they generate: the so-called \emph{Pauli root groups}.  In this
  work we introduce a formalization of such groups, study finiteness and
  infiniteness properties, and precisely determine equality and subgroup
  relations.
\end{abstract}

\maketitle


\section{Introduction}
For the realization of quantum computers, one uses circuits that perform actions
on a single qubit and thus are represented by $2\times2$ unitary matrices.  For
this purpose, often roots of the Pauli matrices are applied.  Together with
operations described by translation matrices, of which the Hadamard matrix forms
a special case, many gate libraries such as the Clifford
group~\cite{DD:03,Got:96}, the Clifford+$T$
group~\cite{Sel:15,PhysRevA.87.032332}, or recently the Clifford+$T_n$
group~\cite{FGKM:15} appear in the literature.  The Clifford group is a finite
group and has applications in so-called stabilizer circuits~\cite{Got:96} while
the Clifford+$T$ group can be used for the exact synthesis of all unitary
matrices~\cite{PhysRevA.87.032332,Sel:15}.

Similar to~\cite{SMD:13} we are generalizing the roots of the Pauli matrices,
denoted $\ur[k]{a}$ (with $k \in \N$ and $a \in \{1,2,3\}$), but call them roots
of the identity, since the square of a Pauli matrix gives the identity matrix.
In the notation of the \emph{identity root}, the integer $a$ refers to one of
the three directions in the Bloch sphere and $k$ is the degree.  Note that $k$
plays the same role as $2n$ in~\cite{FGKM:15}.  Many famous matrices from the
literature are specializations of the identity root, e.g., $X = \ur[2]{1}$,
$Y = \ur[2]{2}$, $Z = \ur[2]{3}$, $S = \ur[4]{3}$, and $T = \ur[8]{3}$.  The
Hadamard operation $H$ can be regarded as a translator between the $X$ and $Z$
direction in the Bloch sphere as we have, e.g., $X = HZH$ and $Z = HXH$.  In
fact, as we will prove we also have $\ur[k]{1} = H\ur[k]{3}H$ and
$\ur[k]{3} = H \ur[k]{1} H$ which gives reason to consider a generalization of
this translation.  Since there are three directions in the Bloch sphere, there
are also three \emph{translation matrices}, denoted $\rho_{ab}$, of which
$H = \rho_{13}$ is a special case.

The above mentioned gate libraries, Clifford group and Clifford+$T$ group, are
equal to the groups $\<S, H>$ and $\<T, H>$, respectively.  That is, a group
generated by an identity root and a translation matrix.  In this paper we will
consider all such groups $\<\ur[k]{a}, \rho_{bc}>$, i.e., the matrix groups
generated by the two matrices $\ur[k]{a}$ and $\rho_{bc}$.  We call them
\emph{Pauli root groups}.  All are countable subgroups of the continuous group
$\mathrm{U}(2)$.  Our contributions are as follows: We introduce notations and
derive elementary properties for such groups (Section~\ref{sec:prelims}).  We
study finiteness and infiniteness properties
(Section~\ref{sec:equality-relation}).  Further, we precisely determine equality
and subgroup relations between two Pauli root groups
(Sections~\ref{sec:equality-relation} and~\ref{sec:subset-relation}).

\section{Preliminaries}
\label{sec:prelims}

The three \emph{Pauli matrices}~\cite{CM:1929} are given by
\begin{equation}
  \sigma_1 = \left(\begin{array}{cr} 0 & 1 \\ 1 & 0 \end{array}\right), \;
  \sigma_2 = \left(\begin{array}{cr} 0 & -\rmi \\ \rmi & 0 \end{array}\right), \;
  \sigma_3 = \left(\begin{array}{cr} 1 & 0 \\ 0 & -1 \end{array}\right).
\end{equation}
The alternate naming~$X=\sigma_1$, $Y=\sigma_2$, and $Z=\sigma_3$ is often used
and we use it whenever we refer to a specific Pauli matrix.

Matrices describing \emph{rotations} around the three axes of the Bloch sphere
are given by
\begin{align*}
  R_a(\theta)
  & =
  \cos{\textstyle\frac{\theta}{2}}I-\rmi\sin{\textstyle\frac{\theta}{2}}\sigma_a,
\end{align*}
where~$a\in \{1,2,3\}$~with~$\theta$ being the rotation angle and~$I$ the
$2\times2$ identity matrix~\cite{NC:2000}.  Each Pauli matrix specifies a half
turn $(180^{\circ})$ rotation around a particular axis up to a global phase,
i.e.,
\[
\sigma_a=e^{\frac{\rmi\pi}{2}}R_a(\pi).
\]
The conjugate transpose of~$R_a(\theta)$ is found by negating the
angle~$\theta$ or by multiplying it by another Pauli matrix~$\sigma_b$ from both
sides, i.e.,
\[
  R^\dagger_a(\theta)=R_a(-\theta)=\sigma_bR_a(\theta)\sigma_b,
\]
where~$a\neq b$. Note that it does not matter which of the two possible
$\sigma_b$ is used.  Since~$\sigma_b$ is Hermitian, we also have
$R_a(\theta)=\sigma_bR_a^\dagger(\theta)\sigma_b$.

\begin{defi}[Identity root]
  Let $\omk=e^{\frac{\rmi2\pi}{k}}$ be a $k^{\rm th}$ root of unity.  Then
  \[
    \ur[k]{a} = \omega_{2k}R_a(\tfrac{2\pi}{k})
  \]
  is referred to as the \emph{identity root} of \emph{direction} $a$ and
  \emph{degree} $k$ where $a\in\{1,2,3\}$ and $k\in\N$.
\end{defi}

We have
\begin{align}
  \label{eq:vka}
  \ur[k]{a} &= \omega_{2k}\left(\cos\left(\tfrac\pi{k}\right)I - \rmi
              \sin\left(\tfrac\pi{k}\right)\sigma_a\right) \nonumber \\
            &= \omega_{2k} \left( \frac{\omega_{2k} + \omega_{2k}^{-1}}{2} I - \frac{\omega_{2k} - \omega_{2k}^{-1}}{2} \sigma_a \right) \nonumber \\
            &= \tfrac{1}{2} \left( (\omega_{2k}^2 + 1) I - (\omega_{2k}^2 - 1) \sigma_a \right) \nonumber \\
            &= \tfrac{1}{2} \left( (1 + \omega_k) I + (1 - \omega_k) \sigma_a \right).
\end{align}
For $k=1$ we have $V_{1,a}=I$ and for $k=2$ we have $V_{2,a} = \sigma_a$.  In
particular, if $a=1$ or $a=3$, then the matrix $\ur[k]{a}$ has entries in
$\Q(\omk)$, the smallest subfield of $\C$ containing $\omk$. If $a=2$, then it
has entries in $\Q(\omk, \rmi)$, the smallest subfield of $\C$ containing both
$\omk$ and $\rmi$.  We also note that
\[
  \ur[k]{a}^\dagger = \tfrac12 \left((1+\omk^{-1})I + (1-\omk^{-1})\sigma_a \right)
\]
and
\[
  \det(V_{k,a}) = \omega_k.
\]

Explicit forms of the three identity roots are
\begin{align*}
  \ur[k]{1} & = \frac12\lmat{1+\omk}{1-\omk}{1-\omk}{1+\omk}, \\
  \ur[k]{2} & = \frac12\lmat{1+\omk}{-\rmi+\rmi\omk}{\rmi-\rmi\omk}{1+\omk},\,\text{and} \\
  \ur[k]{3} & = \lmat100{\omk}.
\end{align*}
Note that the identity roots are related to the identity roots that have been
presented in \cite{SMD:13}.  We have $\pr[k]{a}=\ur[2k]{a}$.

Throughout we will make use of the Levi-Civita symbol: for $a,b,c \in \{1,2,3\}$
we write $\varepsilon_{abc}=(a-b)(b-c)(c-a)/2$, i.e.,
$\varepsilon_{123}=\varepsilon_{231}=\varepsilon_{312}=1$,
$\varepsilon_{321}=\varepsilon_{213}=\varepsilon_{132}=-1$, and $0$ for all
other cases.

\emph{Translation} from one Pauli matrix to another is given by
\begin{equation}
  \label{eq:pauli-to-pauli}
  \sigma_a=\rho_{ab}\sigma_b\rho_{ab} \quad\text{and}\quad
  \sigma_b=\rho_{ab}\sigma_a\rho_{ab},
\end{equation}
where
\begin{equation}
  \label{eq:translation-matrix}
    \rho_{ab}=\rho_{ba}=\tfrac{1}{\sqrt2}(\sigma_a+\sigma_b) 
    =e^{\frac{\rmi\pi}{2}}R_a\left(\tfrac{\pi}{2}\right)
    R_b\left(\tfrac{\pi}{2}\right)
    R_a\left(\tfrac{\pi}{2}\right)
\end{equation}
with~$a\neq b$ are \emph{translation matrices}. Explicit forms for the three
translation matrices are
\begin{align*}
   \rho_{12} &= \osqr\lmat0{1-\rmi}{1+\rmi}0 = \lmat0{\omega_8^7}{\omega_8}0, \\
   \rho_{13} &=H= \osqr\lmat111{-1}, \text{and} \\
   \rho_{23} &= \osqr\lmat1{-\rmi}\rmi{-1}.
\end{align*}
Further, we define~$\rho_{aa}=I$.  Note that~\eqref{eq:pauli-to-pauli}
describes a conjugation since $\rho^{-1}_{ab}=\rho_{ab}$.  It can be extended to
the identity roots, giving
\begin{equation}
  \label{eq:pauli-root-to-pauli-root}
  \ur[k]{a}=\rho_{ab}\cdot\ur[k]{b}\cdot\rho_{ab}\quad\text{and}\quad
  \ur[k]{b}=\rho_{ab}\cdot\ur[k]{a}\cdot\rho_{ab},
\end{equation}
as announced in the introduction, which can be proven using~\eqref{eq:vka}
and~\eqref{eq:translation-matrix}.

On the other hand, if $\varepsilon_{abc} \neq 0$, then
\begin{equation}
  \label{eq:pauli-negate}
  -\sigma_c = \rho_{ab}\sigma_c\rho_{ab}.
\end{equation}
We conclude that conjugation by a translation matrix permutes the set
\[
  \Sigma^\pm = \{ \sigma_1, \sigma_2, \sigma_3,
                  -\sigma_1, -\sigma_2, -\sigma_3 \}
\]
of Pauli matrices and their negatives, in three different ways.  Other types of
permutations are possible: if $\varepsilon_{abc} = 1$, then using~\eqref{eq:vka}
along with $\sigma_a\sigma_b\sigma_c = \rmi\varepsilon_{abc}$ one verifies that
\begin{align}
  \label{eq:transa}
  \sigma_a &= \ur[4]{a} \sigma_a \ur[4]{a}^\dagger, \\
  \label{eq:transb}
  \sigma_c &= \ur[4]{a} \sigma_b \ur[4]{a}^\dagger, \\
  \label{eq:transc}
 -\sigma_b &= \ur[4]{a} \sigma_c \ur[4]{a}^\dagger.
\end{align}
By combining these formulas with~\eqref{eq:pauli-to-pauli}
and~\eqref{eq:pauli-negate} in all possible ways, we obtain 24 permutations of
$\Sigma^\pm$ that are induced by conjugation.  (Fun fact: these correspond to
the rotations of a die with labels $\sigma_1$, $\sigma_2$, $\sigma_3$,
$-\sigma_3$, $-\sigma_2$, and $-\sigma_1$ in place of the numbers $1,\dots,6$.)
The most interesting permutations are described by the following lemma.
\begin{lemma}
  \label{lemma:conjugation}
  If $\eps = 1$, then the conjugation
  \[ U \mapsto \rho_{ab} \ur[4]{a}^\dagger \cdot U \cdot \ur[4]{a}\rho_{ab} \]
  cyclically permutes $\sigma_a$, $\sigma_b$, and $\sigma_c$.  Consequently, it
  also cyclically permutes $\ur[k]{a}$, $\ur[k]{b}$, and $\ur[k]{c}$, as well
  as $\rho_{ab}$, $\rho_{bc}$, and $\rho_{ca}$.
\end{lemma}

\noindent
\begin{proof}
  Using the preceding formulas, the reader verifies that indeed
 \begin{align*}
    \rho_{ab} V_{4,a}^\dagger \cdot \sigma_a \cdot V_{4,a} \rho_{ab} &=  \rho_{ab} ( V_{4,a}^\dagger \sigma_a V_{4,a}) \rho_{ab}  =  \sigma_b \\
    \rho_{ab} V_{4,a}^\dagger \cdot \sigma_b \cdot V_{4,a} \rho_{ab} &=  \rho_{ab} ( V_{4,a}^\dagger \sigma_b V_{4,a}) \rho_{ab}  =  \sigma_c \\
    \rho_{ab} V_{4,a}^\dagger \cdot \sigma_c \cdot V_{4,a} \rho_{ab} &=  \rho_{ab} ( V_{4,a}^\dagger \sigma_c V_{4,a}) \rho_{ab}  =  \sigma_a.
  \end{align*}
  The other statements then follow from Formulas~\eqref{eq:vka},
  \eqref{eq:translation-matrix}, and~\eqref{eq:pauli-root-to-pauli-root}.
\end{proof}

Formulas~\eqref{eq:transa} and~\eqref{eq:transb} lead to the following two
useful corollaries.

\begin{corollary}
\label{lemma:unity-root-as-others}
For $\epsilon_{abc}\ne 0$ we have that
\[
  \ur[k]{b}=
  \begin{cases}
    \ur[4]{c}\cdot\ur[k]{a}\cdot\ur[4]{c}^\dagger
    & \text{if $\epsilon_{abc}=1$,} \\
    \ur[4]{c}^\dagger\cdot\ur[k]{a}\cdot\ur[4]{c}
    & \text{if $\epsilon_{abc}=-1$}. \\
  \end{cases}
\]
\end{corollary}

\begin{corollary}
\label{lemma:translation-from-other-and-roots}
With~\eqref{eq:translation-matrix} one further gets
\[
  \rho_{ac}=
  \begin{cases}
    \ur[4]{a}\cdot \rho_{ab}\cdot \ur[4]{a}^\dagger
    & \text{if $\epsilon_{abc}=1$,} \\
    \ur[4]{a}^\dagger\cdot \rho_{ab}\cdot\ur[4]{a}
    & \text{if $\epsilon_{abc}=-1$.}
  \end{cases}
\]
\end{corollary}

\begin{defi}[Pauli root group]
  A group that is generated by an identity root and a translation matrix
  \[
    P = \< \ur[k]{a}, \rho_{bc} >
  \]
  is called a \emph{Pauli root group} of \emph{degree} $k$.  We have twelve such
  groups.  We distinguish three different group properties.  The group $P$ is
  called \emph{cyclic} if $b=c$ and hence $\rho_{bc}=I$.  In that case we have
  $P=\< \ur[k]{a} >$.  The group $P$ is called \emph{polycyclic}, if
  $\varepsilon_{abc}\neq 0$.  In that case all indices are distinct.  If $P$ is
  neither cyclic nor polycyclic it is called \emph{smooth}.
\end{defi}

It can easily be checked that a smooth Pauli root group has the form
$\< \ur[k]{a}, \rho_{ab}>$ with $a\neq b$.

\begin{remark}
\label{rem:iso}
The isomorphism class of the Pauli root group $\<\ur[k]{a}, \rho_{bc}>$ only
depends on its degree $k$ and on whether the group is cyclic, polycyclic, or
smooth.  Indeed, the corresponding groups are all conjugate to each other.  This
follows immediately from Lemma~\ref{lemma:conjugation}, except in the case of
two smooth Pauli root groups involving the same identity root, in which case one
can use
\[ \ur[4]{a}\<\ur[k]{a},\rho_{ab}>\ur[4]{a}^\dagger = \<\ur[k]{a},\rho_{ac}>, \]
where we assume without loss of generality that $\varepsilon_{abc} = 1$.  To see
this, one uses~\eqref{eq:transa} and~\eqref{eq:transb} along with~\eqref{eq:vka}
and~\eqref{eq:translation-matrix}.
\end{remark}

\section{Equality Relation}
\label{sec:equality-relation}
This section finishes with a theorem that gives a precise description of when
Pauli root groups of the same degree are equal and when not.  The practical
application of the theorem is that gate libraries can easily be exchanged in
case of equality of their respective Pauli root groups.

First we introduce a fact from algebraic number theory.  Recall that a complex
number is called an \emph{algebraic integer} if it is a root of a polynomial
with leading coefficient $1$ and all coefficients in $\Z$.  It is well-known
that the sum and product of two algebraic integers are again algebraic integers
\cite[Thm.\,2.8]{ST:79}.
\begin{lemma}\label{cyclofield}
  The following statements hold.
  \begin{enumerate}
  \item If $\omega_l \in \Q(\omk)$ for some positive integer $l$ then
    $l \mid \lcm(k,2)$.
  \item The field $\Q(\omk)$ is a $\Q$-vector space with basis
    \begin{equation}\label{cyclobasis}
      \{ 1, \omk, \omk^2, \dots, \omk^{\varphi(k) - 1} \},
  \end{equation}
  where
  \[ \varphi(k) = \# \{ c \in \mathbb{Z}\ |\, 1 \leq c \leq k \text{ and }
  \gcd(c,k) = 1\} \]
  is Euler's totient function.
\item The algebraic integers of $\Q(\omk)$ are precisely those elements for
  which all of the coordinates with respect to this basis are integers.
  \end{enumerate}
\end{lemma}
\begin{proof} By \cite[Ex.\,2.3]{Was:82} one has
  $\omega_l \in \set{\pm \omk^i}{i \in \Z }$.  So the first statement follows by
  writing $-1 = \omega_2$ and using the general fact that
  \[ \<\omega_m, \omega_n> = \< \omega_{\lcm(m,n)}>.\]
  The other statements are \cite[Thm.\,2.5 and 2.6]{Was:82}, respectively.
\end{proof}

Before we get to the core of Theorem~\ref{theo:eq}, we need to shed some light
into the properties of the Pauli root groups and will prove some lemmas for
special cases.  All Pauli root groups are countable.  First, we want to exactly
classify for which cases Pauli root groups are finite and for which they are
infinite.
\begin{lemma}
  \label{lem:fin}
  $P=\< \ur[k]{a}, \rho_{bc} >$ is finite if, and only if
  \begin{enumerate}
  \item $P$ is cyclic, or
  \item $P$ is polycyclic, or
  \item $P$ is smooth and $k\in\{1,2,4\}$.
  \end{enumerate}
\end{lemma}

\begin{proof}
  We will prove each case separately:
  \begin{enumerate}
  \item For the cyclic case we have $P=\<\ur[k]{a}>$ and
    $\left(\ur[k]{a}\right)^k=I$.  Thus $P$ is finite of order $k$.
  \item Due to Remark~\ref{rem:iso}, it suffices to assume that $a=3$. Recall
    that
    \[ \ur[k]{3} = \lmat100\omk. \]
    Because
    \[
    \rho_{12}\ur[k]{3}\rho_{12} =
    \lmat0{\omega_8^7}{\omega_8}0
    \lmat100\omk
    \lmat0{\omega_8^7}{\omega_8}0
    = \lmat\omk001,
    \]
    our group $P$ contains
    \[ N = \<\ur[k]{3},\rho_{12}\ur[k]{3}\rho_{12}>
         = \left\{\lmat{\omk^t}00{\omk^s} \middle| 0 \le s,t < k \right\}
         \cong C_k \times C_k,
    \]
    where $C_k$ denotes the cyclic group of order $k$.  This is a strict
    subgroup of $P$, as it does not contain $\rho_{12}$ (which is not a diagonal
    matrix).  We claim that every element of $P$ can be written as
    \begin{equation}
      \label{eq:pform}
      \ur[k]{3}^s(\rho_{12}\ur[k]{3}\rho_{12})^t\rho_{12}^u
    \end{equation}
    with $0 \le s,t < k$ and $u \in \{0,1\}$.  To see this, it suffices to note
    that the form~\eqref{eq:pform} is preserved when multiplied from the right
    by $\ur[k]{3}$ or $\rho_{12}$.  In case of $\rho_{12}$ this is clear.  In
    case of $\ur[k]{3}$ and $u=0$ this follows from commutativity of $N$.  In
    the remaining case of $\ur[k]{3}$ and $u=1$ we have
    \[ \ur[k]{3}^s(\rho_{12}\ur[k]{3}\rho_{12})^t \rho_{12}\ur[k]{3}
       = \ur[k]{3}^s(\rho_{12}\ur[k]{3}\rho_{12})^{t+1} \rho_{12}.
    \]
    So the claim follows, and we conclude that $N$ is an index two (hence
    normal) subgroup and $P$ is the semi-direct product
    $N \rtimes \<\rho_{12}> \cong (C_k \times C_k) \rtimes C_2$.  In particular
    $|P| = 2k^2$.  Because conjugation by $\rho_{12}$ swaps $\ur[k]{3}$ and
    $\rho_{12}\ur[k]{3}\rho_{12}$, the underlying action of $C_2$ on
    $C_k \times C_k$ is by permutation.  So in fact $P$ is the wreath product
    $C_k \mathrel{\wr} C_2$, also known as the \emph{generalized symmetric
      group} $S(k,2)$.
  \item In case of $k=1$ we have the group $\< \rho_{ab}>$ which has 2
    elements.  If $k=2$, then $P = \<\sigma_a, \rho_{ab}>$ is generated by two
    involutions, so its structure is determined by the order of
    $\sigma_a\rho_{ab}$ which is $8$.  Therefore the group is isomorphic to
    $D_8$, the dihedral group with 16 elements.

    If $k=4$, then by Lemma~\ref{lemma:conjugation} we see that $P$ always
    contains $\rho_{12}$, $\rho_{13}$, $\rho_{23}$, $\ur[4]{1}$, $\ur[4]{2}$,
    and $\ur[4]{3}$.  So we have just one concrete group, namely the Clifford
    group, independent of the indices $a$ and $b$.  Using a computer algebra
    package it is easy to confirm the known fact that this group has 192
    elements.  Alternatively, one can consider the normal subgroup of $P$ of
    matrices having determinant $1$.  Because this is a subgroup of
    $\mathrm{SU}(2)$, its elements can be identified with quaternions, using the
    usual representation of the quaternions as $2\times2$ matrices over $\C$.
    Denote by $Q_{48}$ the group of quaternions corresponding to our normal
    subgroup.  Then one can verify that $Q_{48}$ is an instance of the
    \emph{binary octahedral group} 2O, a well-known group of order 48; see
    \cite[\textsection 7.3]{Cox:74}.  Our Pauli root group $P$ can then be
    written as $Q_{48} \rtimes \< \ur[4]{3}> \cong 2O \rtimes C_4$.  Here
    $\ur[k]{3}$ acts on $Q_{48}$ as conjugation by the quaternion $1-\rmi$,
    because the latter is represented by the matrix $(1-\rmi)\ur[4]{3}$.

    Now we investigate the remaining cases. Let $P=\<\ur[k]{a}, \rho_{ab}>$
    be smooth with $k \notin \{1,2,4\}$ and let $U = \ur[k]{a} \rho_{ab}$.
    With the help of~\eqref{eq:vka}, we calculate
\begin{align*}
  U &= \tfrac12
       \left( (1+\omk)I + (1-\omk)\sigma_a \right)
       \osqr
       (\sigma_a + \sigma_b) \\
    &= \tfrac{1}{2\sqrt2}
       \left( (1 + \omk)(\sigma_a + \sigma_b) + (1 - \omk) (I + \sigma_a\sigma_b) \right).
\end{align*}
Note that the Pauli matrices have trace 0, and that the same is true for $\sigma_a \sigma_b$,
being a scalar times a Pauli matrix.
Therefore,
\[
  \Tr(U) = \osqr (1 - \omega_k).
\]
Now if $U$ would have finite order, or in other words if $U^n = I$ for some
$n \geq 1$, then the eigenvalues $\lambda_1, \lambda_2$ of $U$ would also
satisfy $\lambda_i^n = 1$.  In particular they would be algebraic integers.
Therefore also
\begin{equation}\label{eq:square}
  (\lambda_1 + \lambda_2)^2 = \Tr(U)^2 = \tfrac12 - \omk + \tfrac12 \omk^2
\end{equation}
would be an algebraic integer. But it clearly concerns an element of
$\Q(\omk)$. So by Lemma~\ref{cyclofield} the coordinates of~\eqref{eq:square}
with respect to~\eqref{cyclobasis} would have to be integers.

We can now conclude that the coordinates of (\ref{eq:square}) with respect to the basis \eqref{cyclobasis} are
\[ \tfrac{1}{2}, \ -1, \ \tfrac{1}{2}, \ 0, \ 0, \ \dots, \ 0 \]
and therefore we run into a contradiction, except if $\varphi(k) - 1$ is smaller
than~2. The exception happens if and only if $k$ equals 1, 2, 3, 4, or
6. Excluding the cases $k=1$, 2, or 4, we still have to investigate the cases
where $k$ equals 3 or 6. Then $\varphi(k) = 2$ and we can use the identities
$\omega_3^2 = - \omega_3 - 1$ and $\omega_6^2 = \omega_6 - 1$ to rewrite
\[ \tfrac12 - \omega_3 + \tfrac12 \omega_3^2 = -\tfrac32\omega_3
   \quad\text{and}\quad
   \tfrac12 - \omega_6 + \tfrac12 \omega_6^2 = -\tfrac12\omega_6,
\]
respectively, so that the same conclusion follows.
  \end{enumerate}
  As all cases for the Pauli root groups (cyclic, polycyclic, and smooth) have
  been investigated, the `only if' direction is also shown.
\end{proof}

\begin{remark}
  In the smooth $k=4$ case, the Clifford group, the considerations from
  Section~\ref{sec:prelims} show that our group $P$ naturally acts on
  $\Sigma^\pm$ by conjugation.  This gives a surjective homomorphism from $P$ to
  a group with $24$ elements, whose kernel consists of the scalar matrices in
  $P$.  It is not hard to verify that these are precisely
  $(\rho_{13}\ur[k]{3})^{3i} = \omega_8^i I$ for $i = 0,\dots,7$.  This gives
  another way of seeing that $|P| = 8\cdot 24 = 192$.  A similar type of
  reasoning has been made in~\normalfont{\cite{Ozo:08}}.
\end{remark}

\begin{remark}
  The Clifford group $P=\<\ur[4]{a},\rho_{ab}>$ is also naturally isomorphic to
  $\mathrm{GU}(2,9)$, the group of unitary similitudes (sometimes called the
  \emph{general unitary group}) of dimension $2$ over the field $\F_9$ with $9$
  elements.  This is the group of $2\times2$ matrices $U$ over $\F_9$ such that
  $U^\dagger U=\pm I$.  The operation $U \mapsto U^\dagger$ is the conjugate
  transpose, where the conjugate is defined element-wise by
  $\F_9 \to \F_9 : a \mapsto a^3$ (the Frobenius automorphism).  Remark that if
  $\bar\imath \in \F_9$ denotes a square root of $-1 \equiv 2 \bmod 3$, then its
  conjugate is just $\bar\imath^3 = -\bar\imath$.  Then our isomorphism is
  established by
  \[ P \to \mathrm{GU}(2,9) : U \mapsto \overline{U} := U \bmod 3 \]
  where reduction modulo 3 makes sense by reducing both $\rmi$ and $\sqrt2$ to
  $\bar\imath$.  Note that the reduction of a matrix involving $\sqrt2$ is no
  longer unitary, which explains why we end up with unitary similitudes; more
  precisely
  \[ \overline{V}_{4,a}^\dagger\overline{V}_{4,a} = I \quad\text{while}\quad
     \overline{\rho}_{ab}^\dagger\overline{\rho}_{ab} = -I.
  \]
  This shows that the reduction map is a well-defined group homomorphism.  It
  turns out that it is bijective.
\end{remark}

Lemma~\ref{lem:fin} shows that $P$ is only infinite if $P$ is smooth and
$k\notin\{1,2,4\}$.  Using the lemma, we are able to relate the order of a
polycyclic Pauli root group to the order of a smooth Pauli root group when they
have the same degree.

\begin{corollary}
  \label{lem:size-inc-smo}
  Let $P$ be a polycyclic Pauli root group and $Q$ be a smooth Pauli root group
  where both have degree $k$.  Then $|P| = |Q|$, if $k=1$, and $|P|<|Q|$
  otherwise.
\end{corollary}

We will prove next that two polycyclic Pauli root groups of the same degree can
never be equal.  Note that for each degree~$k$ there exist only three polycyclic
Pauli root groups, which are $\<\ur[k]{1},\rho_{23}>$, $\<\ur[k]{2},\rho_{13}>$,
and $\<\ur[k]{3},\rho_{12}>$.

\begin{lemma}
  \label{lem:inc}
  Let $P=\< \ur[k]{a}, \rho_{bc} >$ and $Q=\< \ur[k]{b}, \rho_{ac}>$ be two
  polycyclic groups.  Then, $P\neq Q$.
\end{lemma}

\begin{proof}
  Since both groups are polycyclic, we have $a\neq b$.  We now make a case
  distinction on $k$.  If $k=1$, then $P = \{I,\rho_{bc}\}$ and
  $Q = \{I,\rho_{ac}\}$.  For $k>1$, let us assume that $P=Q$.  Then
  $\ur[k]{b}\in P$ and $R=\<\ur[k]{b},\rho_{bc}>$ is a smooth Pauli root group.
  Since $R\subseteq P$, we imply that $|R|\le |P|$.  However, since $P$ is
  polycyclic and $R$ is smooth, $|R|>|P|$ according to
  Corollary~\ref{lem:size-inc-smo} which is a contradiction and hence our
  assumption must be wrong.
\end{proof}

\begin{lemma}
  \label{lem:dif}
  Let $P=\< \ur[k]{a}, \rho_{bc} >$ and $Q=\< \ur[k]{d}, \rho_{ef} >$ be two
  Pauli root groups that are not both polycyclic, not both cyclic, and also not
  both smooth.  Then $P\neq Q$, unless $k=1$ and $\rho_{bc}=\rho_{ef}$.
\end{lemma}

\begin{proof}
  If $P$ is cyclic, then $Q$ is not cyclic and we have $\rho_{ef}\neq I$, and
  $\rho_{ef}\in Q$ but $\rho_{ef}\notin P$.  It remains to consider the case
  where $P$ is polycyclic and $Q$ is smooth.  If $k=1$, then $P=\{I,\rho_{bc}\}$
  and $Q=\{I,\rho_{ef}\}$ and therefore they are only equal if
  $\rho_{bc}=\rho_{ef}$.  If $k>1$, according to
  Corollary~\ref{lem:size-inc-smo} we have $|P|<|Q|$.
\end{proof}

Now we are investigating special cases in which both Pauli root groups are
smooth.  Since they are infinite most of the cases, they are the most difficult
to consider.  First, we consider the case in which both Pauli root groups share
the same translation matrix.  There are three such cases for each degree.

\begin{lemma}
  \label{lem:sm-same-tm}
  Let $P=\<\ur[k]{a},\rho_{ab}>$ and $Q=\<\ur[k]{b},\rho_{ab}>$ with $a\neq b$.
  Then $P=Q$.
\end{lemma}

\begin{proof}
  Using~\eqref{eq:pauli-root-to-pauli-root}, we have
  $\ur[k]{b}=\rho_{ab}\cdot\ur[k]{a}\cdot\rho_{ab}$ and therefore $Q\le P$.
  Further, we have $\ur[k]{a}=\rho_{ab}\cdot\ur[k]{b}\cdot\rho_{ab}$ and
  therefore $P\le Q$.
\end{proof}

We will now prove equality for general smooth Pauli root groups when their
degree is a multiple of 4.

\begin{lemma}
  \label{lem:sm-same-div}
  Let $P=\<\ur[4k]{a},\rho_{ab}>$ and $Q=\<\ur[4k]{c},\rho_{cd}>$ be both
  smooth.  Then $P=Q$.
\end{lemma}

\begin{proof}
  Assume that we have~$\ur[4k]a$ and~$\rho_{ab}$.  According to
  \eqref{eq:pauli-root-to-pauli-root} one can obtain a second root~$\ur[4k]{b}$
  from~$\ur[4k]{a}$ by multiplying it with~$\rho_{ab}$ on both sides.  We
  obtain~$\ur[4]{a}$ from~$\left(\ur[4k]{a}\right)^k$ such that the third Pauli
  root~$\ur[4k]{c}$ can be retrieved by applying
  Corollary~\ref{lemma:unity-root-as-others}.  Analogously to~$\ur[4]{a}$ one
  can also get~$\ur[4]{b}$ and~$\ur[4]{c}$ from~$\ur[4k]{b}$ and~$\ur[4k]{c}$,
  respectively.  Then, other translation matrices are obtained from
  Corollary~\ref{lemma:translation-from-other-and-roots}.  Hence, we have $Q\le
  P$. Analogously we can show $P\le Q$.
\end{proof}

\begin{lemma}
  \label{lem:sm-neq-div}
  Let $P=\<\ur[k]{a},\rho_{ab}>$ and $Q=\<\ur[k]{a},\rho_{ac}>$ with
  $\varepsilon_{abc}\neq0$.  If $4 \nmid k$, then $P\neq Q$.
\end{lemma}

\begin{proof}
  By Lemma~\ref{lemma:conjugation} we can assume that $c = 2$.  In that case
  $\{a,b\} = \{1,3\}$ and the entries of the matrices $\ur[k]{a}$ and
  $\rho_{ab}$ all lie in $\Q(\omk, \sqrt{2})$.  Therefore, all matrices of $P$
  have entries in $\Q(\omk, \sqrt{2})$.

  We claim that $\rmi \not\in \Q(\omk, \sqrt{2})$. Indeed, because
  $\omega_8 = (1 + \rmi)/\sqrt{2}$ the contrary would imply that
  $\Q(\omk, \sqrt{2}) = \Q(\omk, \omega_8) = \Q(\omega_{\lcm(8,k)})$.  By
  Lemma~\ref{cyclofield} the latter field has dimension $\varphi(\lcm(8,k))$
  over $\Q$.  On the other hand every element of
  $\Q(\omk, \sqrt{2})$ can be written as $x + \sqrt{2}y$ with
  $x, y \in \Q(\omk)$, and so the dimension is at most
  $2 \varphi(k)$ over $\Q$.  We would therefore have
  \[ 4 \varphi(k) = \varphi(\lcm(8,k)) \leq 2 \varphi(k), \]
  where the first equality holds because $4 \nmid k$. This is clearly a
  contradiction.

  Since the matrix $\rho_{ac} = (\sigma_a + \sigma_2)/\sqrt{2}$ has a lower-left
  entry either $(1 + \rmi)/\sqrt{2}$ (case $a=1$) or $\rmi/\sqrt{2}$ (case
  $a=3$), it can therefore not be contained in $\<V_{k,a}, \rho_{ab}>$, which
  finishes the proof.
\end{proof}

\begin{lemma}
  \label{lem:sm-all}
  Let $P=\<\ur[k]{a},\rho_{ab}>$ and $Q=\<\ur[k]{c},\rho_{cd}>$ be two smooth
  Pauli root groups.  Then $P=Q$ if, and only if $\rho_{ab}=\rho_{cd}$ or
  $4\mid k$.
\end{lemma}

\begin{proof}
  If the translation matrices are equal we know that $P=Q$ from
  Lemma~\ref{lem:sm-same-tm}.

  Assume that $\rho_{ab}\neq\rho_{cd}$ and $4\mid k$.  Then we have $P=Q$
  according to Lemma~\ref{lem:sm-same-div}.

  Assume that $\rho_{ab}\neq\rho_{cd}$ and $4\nmid k$.  If $a=c$ we have $P\neq
  Q$ according to Lemma~\ref{lem:sm-neq-div}.  If $a\neq c$ we must have
  $\varepsilon_{abc}\neq 0$.  We consider all cases:

  For the case $a=d$ we compute:
    \[P=\<\ur[k]{a},\rho_{ab}> \neqrf[Lemma]{lem:sm-neq-div}
    \<\ur[k]{a},\rho_{ac}>
    \eqrf[Lemma]{lem:sm-same-tm} \<\ur[k]{c},\rho_{ac}> = Q\]

  For the case $b=c$, we compute:
    \[
    P = \<\ur[k]{a},\rho_{ab}> \eqrf[Lemma]{lem:sm-same-tm}
    \<\ur[k]{b},\rho_{ab}>
    \neqrf[Lemma]{lem:sm-neq-div}
    \<\ur[k]{b},\rho_{bd}>=Q
    \]
  For the case $b=d$ we compute:
    \begin{align*}
    P  &= \<\ur[k]{a},\rho_{ab}> \\
       &\eqrf[Lemma]{lem:sm-same-tm} \<\ur[k]{b},\rho_{ab}> \\
       &\neqrf[Lemma]{lem:sm-neq-div} \<\ur[k]{b},\rho_{bc}> \\
        &\eqrf[Lemma]{lem:sm-same-tm} \<\ur[k]{c},\rho_{bc}>=Q \qedhere
    \end{align*}
\end{proof}

This leads us to the theorem of this section.

\begin{theorem}
\label{theo:eq}
Let
\[ P = \< \ur[k]{a}, \rho_{bc} > \quad\text{and}\quad Q= \< \ur[k]{d}, \rho_{ef} > \]

We have $P=Q$ if, and only if

\begin{enumerate}
\item $k=1$ and $\rho_{bc}=\rho_{ef}$, or
\item $P$ and $Q$ are both cyclic and $a=d$, or
\item $P$ and $Q$ are both smooth and $\rho_{bc}=\rho_{ef}$, or
\item $P$ and $Q$ are both smooth and $4 \mid k$.
\end{enumerate}
\end{theorem}

\begin{proof}
  If $k=1$ and $\rho_{bc}=\rho_{ef}$ we have $P=Q=\{I,\rho_{bc}\}$ which proves
  the first case of the theorem.  Because of Lemma~\ref{lem:dif} all other
  groups with different properties are not equal, where a property is being
  either cyclic, polycyclic, or smooth.  Hence, in the following it is
  sufficient to assume that $P$ and $Q$ have the same property:

  \begin{itemize}
  \item   If both $P$ and $Q$ are cyclic, then for $a \neq d$, we obviously have
    $\<\ur[k]{a}> \neq \<\ur[k]{d}>$ and otherwise $\<\ur[k]{a}> = \<\ur[k]{a}>$,
    which covers the second case of the theorem.
  \item According to Lemma~\ref{lem:inc} two polycyclic groups are not equal.
  \item The third and fourth case follow from Lemma~\ref{lem:sm-all}.
\end{itemize}
All cases have been considered which concludes the proof.
\end{proof}

\section{Subgroup Relation}
\label{sec:subset-relation}
In this section we investigate the relation of Pauli root groups of different
degree but the same direction in the identity root and translation matrix.
First we prove a lemma that will be used in the theorem.

\begin{lemma}
  \label{lemma:uroot-exp-divisor}
  Let $k$ and $d$ be natural numbers.  Then we have
  \[ V_{dk,a}^d = V_{k,a}. \]
\end{lemma}

\begin{proof}
  By Lemma~\ref{lemma:conjugation}, it suffices to prove this for $a=3$.  Then
  the statement is equivalent to
  \[
  {\lmat100{\omega_{dk}}}^d = \lmat100{\omk},
  \]
  which is immediately seen to hold.
\end{proof}

This leads us to the theorem of this section.

\begin{theorem}
  Let $P=\<\ur[k]{a},\rho_{bc}>$ and $Q=\<\ur[l]{a},\rho_{bc}>$ be two Pauli
  root groups.  Then $P\le Q$ if, and only if $k\mid l$.
\end{theorem}
\begin{proof}
  The `if' part follows from Lemma~\ref{lemma:uroot-exp-divisor}. As for the
  `only if' part, assume that $P \leq Q$.  By taking determinants we find that
  $\<\omk, -1> \leq \<\omega_l, -1>$ as subgroups of $\C^\times$. We conclude
  that $\lcm(k,2) \mid \lcm(l,2)$.
  \begin{itemize}
    \item If $l$ is even it follows that $k \mid l$, as desired.
   \item If $l$ is odd then $k \mid 2 l$ and we can proceed as follows. Recall that $\rho_{bc} = \frac{1}{\sqrt{2}} (\sigma_b + \sigma_c)$, so we can write
   \begin{equation} \label{largerinclusion}
      V_{k,a} = \frac{1}{\sqrt{2}^s} W
   \end{equation}
   for some $s \in \Z$ and some $W \in \<V_{l,a}, \sigma_b + \sigma_c>$.  This
   implies that
   \[ \frac{1}{\sqrt{2}^s} I = V_{k,a} W^{-1}\]
   has entries in the field
   $\Q(\omk, \omega_{l}, \rmi) \subset \Q(\omega_{4l})$.  The latter field does
   not contain $\sqrt{2}$, for otherwise it would contain
   $\omega_8 = (1 + \rmi) / \sqrt{2}$, which is impossible by
   Lemma~\ref{cyclofield}, because $8 \nmid 4l$ (recall that $l$ is odd).  We
   conclude that $s$ must be even.  On the other hand, by taking determinants
   and using the fact that $s$ is even, (\ref{largerinclusion}) also implies
   that
   $\omega_k \in \langle 1/4, \omega_l, -2 \rangle = \langle \omega_l, -2 \rangle$.
   We conclude that $k \mid l$.
\end{itemize}
The lemma follows.
\end{proof}

\section{Conclusions}
In this paper we have introduced Pauli root groups, which are generated by a
root of the identity matrix and a translation matrix.  The Clifford group and
the Clifford+$T$ group are special instances of the Pauli root groups.  We have
shown properties of these groups and precisely determined equality and
containment relations.  This has useful implications for quantum computing since
now scenarios can be found in which gates can be interchanged without changing
the unitary matrices that can be generated with them.

\section*{Acknowledgments}
We thank Tom De Medts, Aaron Lye, Philipp Niemann, Erik Rijcken, Michael
Kirkedal Thomsen, and Jasson Vindas for helpful discussions.  The authors
acknowledge the support by the European COST Action IC~1405 `Reversible
Computation.'

\bibliography{library}

\end{document}